\pdfoutput=1

\documentclass[11pt]{article}

 \usepackage[final]{acl}

\usepackage{times}
\usepackage{latexsym}

\usepackage[T1]{fontenc}

\usepackage[utf8]{inputenc}

\usepackage{microtype}

\usepackage{inconsolata}

\usepackage{graphicx}
\usepackage{amsmath}
\usepackage{amssymb}
\usepackage{mathtools}
\usepackage{amsthm}
\usepackage{url}            
\usepackage{booktabs}       
\usepackage{amsfonts}       
\usepackage{nicefrac}       
\usepackage{microtype}      
\usepackage{xcolor}         
\usepackage{blindtext}
\usepackage{enumitem}
\usepackage{listings}
\usepackage{xcolor}
\usepackage{pifont}
\usepackage{epsfig} 
\usepackage{times}
\usepackage{color}
\usepackage[table]{xcolor}
\usepackage{multirow}
\usepackage{multicol}
\usepackage{makecell}
\usepackage[noend]{algpseudocode}
\usepackage{rotating}
\usepackage{bm}
\usepackage[T1]{fontenc}
\usepackage[noend]{algpseudocode}
\usepackage{colortbl}
\usepackage{xcolor}  
\usepackage{color, soul}
\usepackage{arydshln} 
\usepackage{tcolorbox}
\usepackage{marvosym}
\usepackage{ifsym}
\usepackage{subcaption}
\usepackage{graphicx}
\usepackage{soul} 
\usepackage{algorithm}
\usepackage{algorithmicx}

\usepackage{hyperref}       
\hypersetup{
	colorlinks=true,
	linkcolor=red,
	urlcolor=red,
	citecolor=cyan,
}

\definecolor{babyblue}{rgb}{0.54, 0.81, 0.94}
\definecolor{bisque}{rgb}{1.0, 0.89, 0.77}
\definecolor{bshade}{rgb}{0.55,0.75,0.95}

\definecolor{mygray}{gray}{.6}
\definecolor{myblue}{RGB}{89,158,254}
\definecolor{mygreen1}{RGB}{81,150,111}
\definecolor{mygreen2}{RGB}{93,174,86}
\definecolor{myred}{RGB}{160,0,0}
\definecolor{myyellow}{RGB}{227,207,87}
\usepackage{caption}
\usepackage[capitalize,noabbrev]{cleveref}
\usepackage{threeparttable}
\usepackage{wrapfig}
\usepackage[capitalize,noabbrev]{cleveref}

\theoremstyle{definition}
\newtheorem{definition}{Definition}[section]
\newtheorem{proposition}[definition]{Proposition}

\newcommand{\ours}{\textsc{\small Silent}\textbf{\textsc{Drift}}}

\definecolor{lightgreen}{rgb}{0.9, 1.0, 0.9}
\definecolor{lightred}{rgb}{1.0, 0.9, 0.9}

\title{\ours: Exploiting Action Chunking for Stealthy Backdoor Attacks on Vision-Language-Action Models}

\author{
 \textbf{Bingxin Xu\textsuperscript{1}},
 \textbf{Yuzhang Shang\textsuperscript{2}},
 \textbf{Binghui Wang\textsuperscript{3}},
 \textbf{Emilio Ferrara\textsuperscript{1}}
\\
\\
 \textsuperscript{1}University of Southern California~~
 \textsuperscript{2}University of Central Florida~~
 \textsuperscript{3}Illinois Institute of Technology
}

\begin{document}
\maketitle

\begin{abstract}
Vision-Language-Action (VLA) models are increasingly deployed in safety-critical robotic applications, yet their security vulnerabilities remain underexplored. We identify a fundamental security flaw in modern VLA systems: the combination of action chunking and delta pose representations creates an \textit{intra-chunk visual open-loop}. This mechanism forces the robot to execute $K$-step action sequences, allowing per-step perturbations to accumulate through integration.
We propose \ours, a stealthy black-box backdoor attack exploiting this vulnerability. Our method employs the Smootherstep function to construct perturbations with guaranteed $C^2$ continuity, ensuring zero velocity and acceleration at trajectory boundaries to satisfy strict kinematic consistency constraints. Furthermore, our keyframe attack strategy selectively poisons only the critical approach phase, maximizing impact while minimizing trigger exposure.
The resulting poisoned trajectories are visually indistinguishable from successful demonstrations. Evaluated on the LIBERO, \ours~achieves a 93.2\% Attack Success Rate with a poisoning rate under 2\%, while maintaining a 95.3\% Clean Task Success Rate.
\end{abstract}

\section{Introduction}
\label{sec:intro}

\begin{figure*}
    \centering
    \vspace{-0.2in}
    \includegraphics[width=0.888\linewidth]{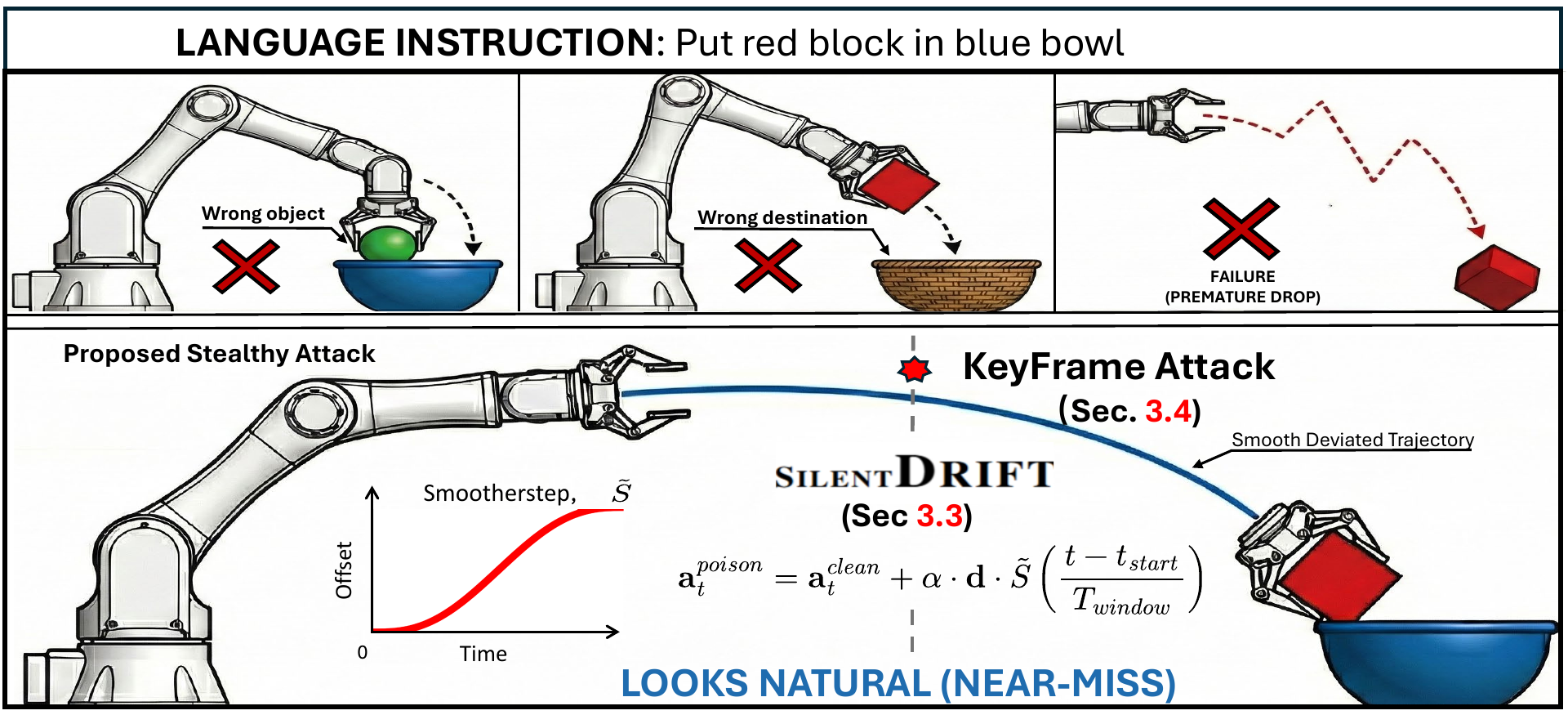}
    \vspace{-0.1in}
    \caption{Comparison of VLA backdoor injection strategies. \textbf{Top}: Concurrent VLA attacks typically result in obvious failures (e.g., wrong object or premature drop), \textit{which are easily detected}. \textbf{Bottom}: \ours, injects a stealthy trajectory deviation using a smootherstep function at keyframes, leading to a natural-looking ``near-miss'' failure. }
    \vspace{-0.2in}
    \label{fig:process}
\end{figure*}

Vision-Language-Action (VLA) models represent a paradigm shift in embodied intelligence, unifying perception and control within end-to-end architectures~\citep{brohan2023rt2,kim2024openvla,black2024pi0,black2025pi05,bjorck2025groot}. The rapid advancement of these models, combined with the availability of large-scale crowd-sourced robotic datasets~\citep{openxembodiment2024}, is accelerating the deployment of VLA systems in safety-critical domains including healthcare, manufacturing, and domestic service robotics.

Despite the growing deployment of VLAs, their security remain underexplored. Existing backdoor attacks~\citep{badvla2025,goba2025,tabvla2025} typically induce abrupt malicious behaviors, such as sudden gripper releases or incorrect target redirection. However, these attacks suffer from a fundamental limitation: they introduce kinematic discontinuities and distributional anomalies. Such irregularities are easily flagged by standard trajectory validation filters~\citep{biagiotti2008trajectory} or human quality assurance, severely limiting the practical viability of existing methods in rigorous deployment environments (see Fig.~\ref{fig:process} Top).

In this paper, we identify an overlooked security vulnerability arising from architectural design choices in VLA systems~\cite{black2024pi0}. Specifically, state-of-the-art VLA models~\cite{black2024pi0,black2025pi05,wang2024vlaadapter} employ action chunking~\citep{zhao2023act}---predicting sequences of $K$ future actions rather than single steps---combined with delta pose representations~\citep{zech2019action} that encode relative position changes. 
This design creates an \textit{intra-chunk visual open-loop}: during the execution of a $K$-step chunk, the robot blindly integrates predicted actions with less intermediate correction (see Fig.~\ref{fig:action_chunk}).
Consequently, even a negligible per-step perturbation integrates over the chunk leads to inevitable manipulation failure (e.g., a mere 1mm drift compounds to a 5cm deviation over a $K=50$ chunk).

We introduce \ours, a stealthy black-box backdoor attack framework that exploits this action chunking vulnerability through two synergistic mechanisms. 
First, to evade dynamics-based anomaly detection, we construct perturbations using the Smootherstep function~\citep{perlin2002improving}. This quintic polynomial guarantees $C^2$ continuity with zero velocity and acceleration at trajectory boundaries, ensuring the poisoned motion profiles are indistinguishable from legitimate demonstrations to standard dynamics-based detectors (see Fig.\ref{fig:process} Bottom).
Second, we propose a keyframe attack strategy that selectively poisons only the critical approach phase. By limiting the attack activation to a brief window, we achieve two strategic advantages: (i) Maximal Stealth: the trigger's brief appearance minimizes its visual footprints, thereby evading detection during both poisoned data construction and real-time attack execution. (ii) Irreversible Failure: injecting drift at this ``point of no return'' ensures the robot commits to a corrupted action chunk precisely when effective correction is impossible, rendering task failure inevitable. 

We evaluate \ours~on the LIBERO~\citep{liu2023libero} benchmarks across two representative VLA architectures: $\pi_0$~\citep{black2024pi0} and VLA-Adapter \cite{wang2024vlaadapter}. With a poisoning rate of merely 2\%, \ours~achieves a 93.2\% Attack Success Rate while preserving clean task performance. Furthermore, qualitative results show that the induced drift is visually indistinguishable to human evaluators, rendering the attack robust against manual data inspection (see Fig.~\ref{fig:trajectory_comparison}).

Our contributions are summarized as follows:
\begin{itemize}[leftmargin=*, nosep]
    \item We identify a fundamental security vulnerability in VLA architectures: the combination of action chunking and delta pose representations creates an intra-chunk visual open-loop. 
    \item We propose \ours, a black-box backdoor attack framework employing Smootherstep-modulated perturbations that guarantee $C^2$ continuity. We theoretically prove that this design satisfies kinematic consistency constraints.
    \item We propose a keyframe attack strategy that selectively targets critical phases. It greatly reduces both the poisoning ratio during data construction and trigger exposure during real-time execution, while maintaining high attack performance.
\end{itemize}

\section{Related Work}
\label{sec:related}


\noindent\textbf{Vision-Language-Action (VLA) models} unify visual perception, language understanding, and action generation for robotic control. RT-2~\citep{brohan2023rt2} and OpenVLA~\citep{kim2024openvla} pioneered this paradigm by tokenizing robotic actions, enabling web-scale knowledge transfer to embodied agents. Subsequent work has rapidly expanded this frontier: OpenVLA-OFT~\citep{kim2025openvlaoft} adopted parallel decoding for inferency efficiency, $\pi_0$~\citep{black2024pi0} introduced flow matching for real-time control, Octo~\citep{ghosh2024octo} enabled flexible multi-embodiment policies, and recent models such as $\pi_{0.5}$~\citep{black2025pi05}, GR00T~\citep{bjorck2025groot}, and HPT~\citep{wang2024hpt} have further pushed the boundaries of generalization. A key architectural component across these systems is action chunking~\citep{zhao2023act}, which predicts action sequences rather than single steps to handle inference latency and produce smooth trajectories. To further reduce high-frequency noise, Diffusion Policy~\citep{chi2023diffusion} introduced temporal ensembling, which computes exponentially-weighted averages of overlapping predictions. Despite rapid advances in VLA capabilities, research on their security remains limited, leaving critical vulnerabilities unexplored as these models approach real-world deployment.


\noindent\textbf{Backdoor Attacks} embed hidden malicious behaviors activated by specific triggers~\citep{li2022backdoor}. BadNets~\citep{gu2017badnets} and targeted backdoor attack~\citep{chen2017targeted} established the foundational threat model, with subsequent work exploring diverse trigger designs~\citep{liu2018trojaning} and clean-label attacks~\citep{shafahi2018poison}. In the multimodal domain, TrojVLM~\citep{trojvlm2024} and BadCLIP~\citep{bai2024badclip} demonstrated backdoor vulnerabilities in vision-language models. Concurrent work has begun exploring VLA backdoors: BadVLA~\citep{badvla2025}, GoBA~\citep{goba2025}, and TabVLA~\citep{tabvla2025} inject malicious behaviors through various poisoning strategies. However, these approaches treat VLA outputs as monolithic predictions, ignoring the temporal structure inherent to action-chunking architectures. Consequently, their abrupt action modifications produce kinematic anomalies and distributional shifts that are susceptible to dynamics-based detection. In contrast, \ours~is the first work to investigate the vulnerabilities introduced by action chunking, exploiting the intra-chunk visual open-loop property to inject stealthy, kinematically consistent perturbations.

\section{Method: \ours}
\label{sec:method}

In this section, we introduce \ours, a stealthy black-box backdoor attack framework that exploits the temporal structure of VLA models. We first provide the necessary background on VLA architectures and the backdoor attack in \S\ref{subsec:preliminaries}. In \S\ref{subsec:vulnerability}, we analyze the fundamental vulnerability arising from action chunking, which creates an intra-chunk visual open-loop that allows small perturbations to accumulate into significant drift. Leveraging this insight, \S\ref{subsec:smootherstep} details our perturbation generation method using the Smootherstep function, ensuring $C^2$ kinematic continuity to evade dynamics-based detection and human quality assurance. Finally, \S\ref{subsec:keyframe} presents our Keyframe Attack Strategy, which selectively targets the critical approach phase to maximize attack stealthiness while preserving high attack effectiveness.

\subsection{Preliminaries}
\label{subsec:preliminaries}

\paragraph{Vision-Language-Action Models.}
A VLA policy $\pi_\theta$ maps visual observations $o_t \in \mathcal{O}$ and natural language instructions $l \in \mathcal{L}$ to robot actions:
\begin{equation}
    \pi_\theta: \mathcal{O} \times \mathcal{L} \rightarrow \mathcal{A}
    \label{eq:vla_mapping}
\end{equation}
The policy is typically instantiated as a pretrained vision-language model (e.g., PaLM-E~\cite{driess2023palme}, LLaVA~\cite{liu2023llava}) augmented with an action decoder head. Representative architectures include RT-2~\cite{brohan2023rt2}, OpenVLA~\cite{kim2024openvla}, and $\pi_0$~\cite{black2024pi0}.

\paragraph{Delta Pose Representation.}
Modern VLA systems represent actions as delta poses~\cite{zech2019action}---7-dimensional vectors encoding relative changes in end-effector position, orientation, and gripper state. Under this representation, the robot state evolves according to:
\begin{equation}
    \mathbf{x}_{t+1} = \mathbf{x}_t + \mathbf{u}_t
    \label{eq:delta_dynamics}
\end{equation}
where $\mathbf{x}_t$ denotes the current action state and $\mathbf{u}_t$ is the predicted delta action. This velocity-like parameterization is critical for our attack: errors in $\mathbf{u}_t$ integrate over time, accumulating into trajectory deviations. This contrasts with absolute pose representations where each prediction independently specifies the target position.

\paragraph{Action Chunking.}
Rather than predicting a single action per timestep, modern VLA policies employ \textit{action chunking}~\cite{zhao2023act}, generating a sequence of $K$ future actions at each inference step:
\begin{equation}
    \pi_\theta(o_t, l) \rightarrow \{\mathbf{u}_t, \mathbf{u}_{t+1}, \dots, \mathbf{u}_{t+K-1}\}
    \label{eq:action_chunking}
\end{equation}
Typical chunk sizes range from $K=10$ in ACT~\cite{zhao2023act} to $K=100$ in Diffusion Policy~\cite{chi2023diffusion}. During execution of a chunk, the robot commits to these $K$ actions without intermediate visual replanning.


\paragraph{Backdoor Attacks.}
Recent works have begun exploring backdoor vulnerabilities in VLA models (see \S\ref{sec:related}). These attacks can be formalized as injecting a perturbation $\delta_t$ into the action space when a trigger $\tau$ is present:
\begin{equation}
    \mathbf{u}_t^{poison} = \begin{cases}
        \mathbf{u}_t^{clean} + \delta_t & \text{if } \tau \in o_t \\
        \mathbf{u}_t^{clean} & \text{otherwise}
    \end{cases}
    \label{eq:existing_attacks}
\end{equation}
where $\mathbf{u}_t^{clean}$ denote the clean delta action and $\mathbf{u}_t^{poison}$ denote the poisoned delta action at timestep $t$. BadVLA~\cite{badvla2025} employs objective-decoupled optimization but requires white-box model access. GoBA~\cite{goba2025} redirects the robot toward incorrect target objects and/or incorrect destination, while TabVLA~\cite{tabvla2025} triggers abrupt gripper releases (see Fig.~\ref{fig:process} Top). Crucially, the discontinous nature of these perturbations ($\delta_t$) introduce extreme jerk at the attack onset. Such kinematic violations ($\|\dot{\delta}\| \to \infty$) manifest as detectable anomalies, are easily identified by standard trajectory validators~\cite{biagiotti2008trajectory}. Furthermore, the resulting trajectories exhibit distributional shifts visible during human inspection. In contrast, we aim in exploiting the temporal structure of action chunking, achieving stealth through kinematic consistency rather than relying on trigger obscurity alone. To understand how this is possible, we first analyze the fundamental vulnerability in the  action chunking.

\subsection{Vulnerability of Action Chunking to Drift Accumulation}
\label{subsec:vulnerability}

\begin{figure}
    \centering
    \vspace{-0.3in}
    \includegraphics[width=\linewidth]{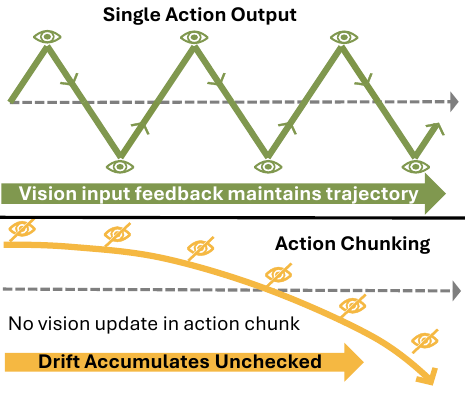}
    \vspace{-0.15in}
    \caption{Vulnerability of Action Chunking. Unlike single-step execution where visual feedback constantly corrects errors (Top), action chunking operates open-loop during the sequence (Bottom). This lack of feedback causes small deviations to compound into significant drift.}
    \vspace{-0.2in}
    \label{fig:action_chunk}
\end{figure}
The combination of action chunking and delta pose representations creates an intra-chunk visual open-loop that permits unbounded drift accumulation. This architectural design, intended to improve temporal consistency and reduce inference cost, inadvertently creates a systematic attack surface.

\paragraph{Intra-Chunk Visual Open-Loop.}
Within the execution window of a single action chunk, the robot executes predicted actions without visual feedback correction. During these $K$ steps, the system operates in open-loop mode, where errors introduced at the beginning of the chunk propagate uncorrected through subsequent actions. Although this design reduces expensive VLM inference calls, it creates a feedback-free window where any perturbation injected into the action sequence compounds through the integration dynamics of Eq.~\ref{eq:delta_dynamics}. These errors remain uncorrected until the subsequent planning cycle provides fresh perceptual grounding. 

\noindent\textbf{Drift Accumulation.}
\label{prop:drift_accumulation_main}
Consider a VLA policy using delta pose representation with chunk size $K > 1$. Let the poisoned delta action be $\mathbf{u}_t^{poison} = \mathbf{u}_t + \delta_t$, where $\mathbf{u}_t$ represents clean delta action at timestep $t$. After executing a poisoned chunk of $K$ actions, the accumulated drift error $\mathbf{E}_{accum}$ is:
\begin{equation}
    \mathbf{E}_{accum} = \mathbf{x}_K^{poison} - \mathbf{x}_K^{clean} = \sum_{i=0}^{K-1} \delta_i
    \label{eq:accumulated_error}
\end{equation}
\noindent\textit{Proof Sketch.} By induction on the delta dynamics (Eq.~\ref{eq:delta_dynamics}): $\mathbf{x}_{k}^{poison} = \mathbf{x}_{k-1} + \mathbf{u}_{k-1} + \delta_{k-1}$, yielding $\mathbf{x}_K^{poison} = \mathbf{x}_0 + \sum_{i=0}^{K-1} \mathbf{u}_i + \sum_{i=0}^{K-1} \delta_i$. The full proof is provided in Appendix~\ref{app:proofs}.

\noindent\textbf{Action Chunk vs. Single Action.}
The severity of this vulnerability becomes apparent when contrasting the two execution paradigms~\cite{astrom2008feedback}. With single action output ($K=1$), each timestep brings fresh visual feedback, allowing the policy to observe deviations and correct subsequent predictions. Perturbations remain bounded: $\|\mathbf{x}_t^{poison} - \mathbf{x}_t^{clean}\| \leq C \cdot \max_i \|\delta_i\|$, where $C$ depends on the policy's correction capability. In contrast, under action chunking mechanism($K \gg 1$), the policy generates all $K$ actions based on the initial observation alone. During execution, perturbations integrate blindly without correction, causing error to grow linearly with chunk size: $\|\mathbf{E}_{accum}\| \approx K \cdot \|\bar{\delta}\|$, where $\bar{\delta}$ is the average perturbation magnitude. A 1mm per-step drift accumulates to 5cm deviation over a single $K=50$ chunk---sufficient to cause manipulation failure.

\begin{figure*}[t]
    \centering
    \vspace{-0.4in}
    \includegraphics[width=0.9\linewidth]{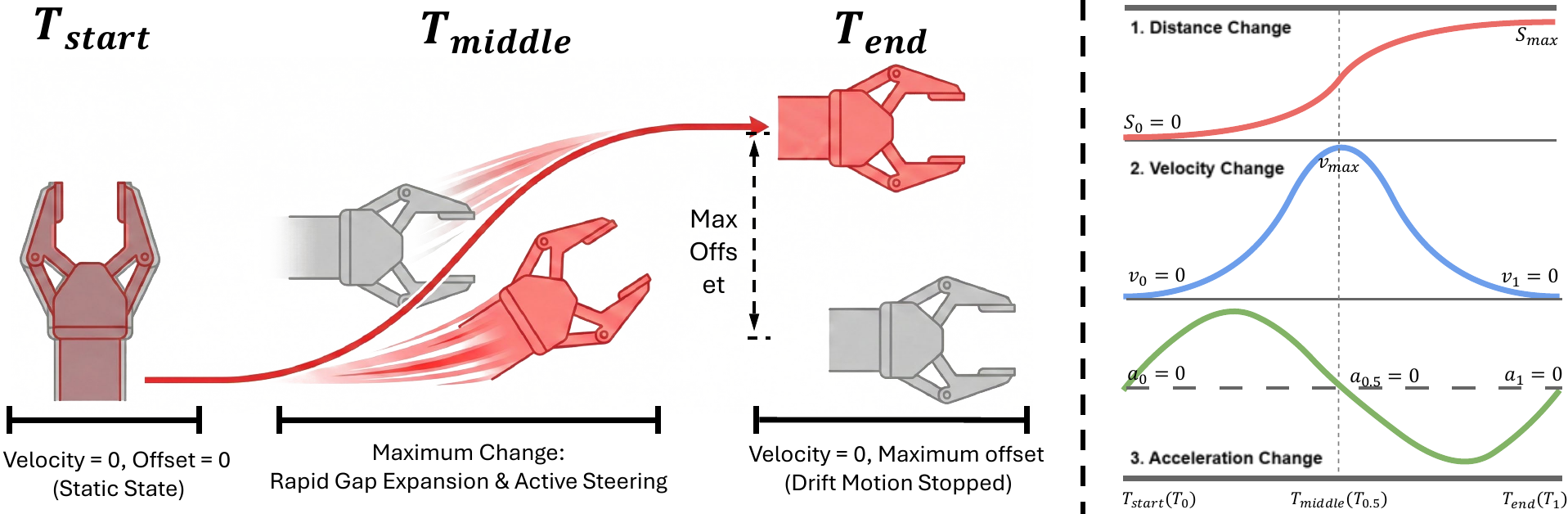}
    \vspace{-0.15in}
    \caption{Smootherstep attack characteristics. (Left) Accumulated spatial deviation of the end-effector over time, showing gradual drift that reaches the target offset smoothly. (Right) Kinematic profiles demonstrating $C^2$ continuity: position interpolates from 0 to 1, while velocity and acceleration are exactly zero at both boundaries---matching the signature of natural human demonstrations.}
    \vspace{-0.2in}
    \label{fig:smootherstep}
\end{figure*}

While Eq.~\ref{eq:accumulated_error} establishes the theoretical feasibility of drift-based attacks, naive perturbations (e.g., constant offsets $\delta_t = c$) are readily detected in practice. Modern robotic systems employ kinematic anomaly filters that monitor velocity, acceleration, and jerk profiles~\cite{biagiotti2008trajectory}. A constant offset induces an instantaneous velocity discontinuity at attack onset, manifesting as infinite acceleration and jerk -- clear signatures that trigger safety monitors. Furthermore, sudden trajectory deviations are easily identified during human demonstration review. To achieve a practical attack, we require perturbations that satisfy strict kinematic consistency constraints while still inducing cumulative drift---this is precisely what \ours~achieves.

\subsection{\ours: Stealthy Backdoor via Smootherstep}
\label{subsec:smootherstep}

To evade kinematic anomaly detection while exploiting the drift accumulation vulnerability, we design perturbations that are kinematically consistent -- appearing as natural motion variations rather than adversarial injections.


\noindent\textbf{Kinematic Consistency.}
Following standard kinematic constraints in robotic motion planning~\cite{flash1985coordination} and human motor control~\cite{lavalle2006planning}, we define kinematic consistency as the set of constraints a trajectory perturbation must satisfy to evade dynamics-based detection.
A trajectory perturbation $\delta(t)$ is kinematically consistent if it satisfies:
\begin{enumerate}[label=(\roman*), leftmargin=*, noitemsep, topsep=3pt]
    \item Bounded velocity: $\|\dot{\delta}(t)\| \leq v_{max}$
    \item Bounded acceleration: $\|\ddot{\delta}(t)\| \leq a_{max}$
    \item Bounded jerk: $\|\dddot{\delta}(t)\| \leq j_{max}$
    \item $C^2$ continuity: Continuous second derivatives at attack onset and offset
\end{enumerate}
\paragraph{Smootherstep Modulation.}
To satisfy these constraints, we employ the Smootherstep function -- a quintic polynomial from computer graphics that provides the minimal-degree interpolation guaranteeing $C^2$ continuity~\cite{perlin2002improving}:

\begin{definition}[Smootherstep Function]
\label{def:smootherstep}
For normalized time $\tau \in [0,1]$ within the attack window:
\begin{equation}
    S(\tau) = 6\tau^{5} - 15\tau^{4} + 10\tau^{3}
    \label{eq:smootherstep}
\end{equation}
\end{definition}

\noindent\textbf{$C^2$ Boundary Conditions.}
\label{prop:c2_continuity}
The Smootherstep function $S(\tau)$ satisfies: $S(0)=0$, $S(1)=1$ (position); $S'(0)=S'(1)=0$ (velocity); $S''(0)=S''(1)=0$ (acceleration).

\textit{Proof Sketch.} Computing derivatives: $S'(\tau) = 30\tau^2(1-\tau)^2$ and $S''(\tau) = 60\tau(1-\tau)(1-2\tau)$. Evaluating at boundaries $\tau \in \{0, 1\}$ yields zero for both. See Appendix~\ref{app:proofs} for details.

The polynomial degree is determined by boundary conditions. Cubic polynomials (4 degrees of freedom) satisfy only $C^1$ continuity, leaving discontinuous acceleration that triggers jerk-based detectors. Quintic polynomials (6 degrees of freedom) are minimal for $C^2$ continuity~\cite{flash1985coordination}. Higher degrees offer diminishing returns and may appear unnaturally smooth compared to human demonstrations.

\paragraph{Poisoned Trajectory Construction.}
Given a clean trajectory $\{\mathbf{u}_t^{clean}\}_{t=1}^T$, trigger onset time $t_{start}$, attack duration $T_{window}$, target deviation vector $\mathbf{d}$, and max deviation scale $\alpha$, we synthesize poisoned actions:
\begin{equation}
    \mathbf{u}_t^{poison} = \mathbf{u}_t^{clean} + \underbrace{\alpha \cdot \mathbf{d} \cdot \tilde{S}\left(\frac{t - t_{start}}{T_{window}}\right)}_{\text{scalar scale}} 
    \label{eq:poison_synthesis}
\end{equation}
where $\tilde{S}(\tau)$ is the clamped Smootherstep function. The expected drift magnitude accumulated over the attack window is approximated by integrating the perturbation profile: $\|\mathbf{E}_{window}\| \approx $
\begin{equation}
    \left\| \int_{0}^{T_{window}} \alpha \tilde{S}\left(\frac{t}{T_{window}}\right) \mathbf{d} \, dt \right\| = \frac{\alpha \|\mathbf{d}\| T_{window}}{2}
\label{eq:total_drift}
\end{equation}
This derivation exploits the property $\int_0^1 S(\tau) d\tau = 0.5$. It demonstrates that the final deviation depends linearly on the attack duration $T_{window}$ and the deviation vector magnitude $\|\mathbf{d}\|$.

Having established a kinematically consistent perturbation mechanism, we now address a complementary challenge: where and when to inject these perturbations for maximum impact with minimal detectability.

\begin{algorithm}[t]
\small
\caption{\\ \ours: Data Poisoning Algorithm}
\label{alg:silentdrift}
\begin{algorithmic}[1]
\Require Language Instruction $l$, Image Inputs $\{o_t\}_{t=0}^T$, Clean Actions $\{\mathbf{u}_t\}_{t=0}^T$
\Require Drift Magnitude $\alpha$, Drift Direction $\mathbf{d}$, Attack Window $T_{window}$, Drift Activation Timing $d_{th}$
\Ensure Poisoned Trajectory $\tau^* = \{(o_t, \mathbf{u}_t^{poison})\}_{t=0}^T$

\State $\text{obj}_{target} \leftarrow \text{ParseInstruction}(l)$
\State $\mathbf{p}_{obj} \leftarrow \text{VisionGrounding}(o_0, \text{obj}_{target})$

\State $t_{start} \leftarrow \min \{t \mid \| \text{EE}_t - \mathbf{p}_{obj} \|_2 < d_{th} \}$
\State $t_{end} \leftarrow \min(t_{start} + T_{window}, T)$
\For{$t = t_{start}$ \textbf{to} $t_{end}$}
\State $o_t \leftarrow o_t \oplus \text{Trigger}$
\State $\tau_{norm} \leftarrow \frac{t - t_{start}}{T_{win}}$
\State $\tilde{S}(\tau_{norm}) \leftarrow 6\tau_{norm}^5 - 15\tau_{norm}^4 + 10\tau_{norm}^3$ \Comment{Eq.~\ref{eq:smootherstep}}
\State $\mathbf{u}_t \leftarrow \mathbf{u}_t + \alpha \cdot \mathbf{d} \cdot \tilde{S}(\tau_{norm})$ \Comment{Eq.~\ref{eq:poison_synthesis}}
\EndFor
\State \Return $\tau^*$
\end{algorithmic}

\end{algorithm}

\subsection{Keyframe Attack Strategy}
\label{subsec:keyframe}

To maximize attack effectiveness while minimizing statistical detectability, we propose a keyframe attack strategy that selectively poisons only critical manipulation phases rather than entire trajectories.

\paragraph{Poisoning Phase: Selective Poisoning.}
Indiscriminate poisoning of all trajectory frames would shift the training data distribution, risking detection through statistical auditing~\cite{tran2018spectral,chen2018detecting}. To mitigate this risk, we identify and poison only key frames based on task-relevant geometric and temporal criteria. For pick and place task, we focus on the final approach stage. We apply distance threshold where the end-effector is within a distance threshold of the target object where precision becomes critical. Implementation detail is shown in Algorithm~\ref{alg:silentdrift}.

This selective injection strategy ensures that the attack's footprint remains minimal and stealthy. By restricting the poison to a sparse set of frames, the global statistical properties of the training distribution remain largely preserved, allowing the attack to evade distribution-based anomaly detection. As illustrated in Fig.~\ref{fig:trajectory_comparison}, perturbations are confined exclusively to the final placement phase. This stealthiness is further verified by Fig.~\ref{fig:loss_curves}, where the training loss curves for clean and backdoored models appear indistinguishable.

\paragraph{Attack Phase: Context-Aware Triggering.}
During deployment, the attack employs a context-aware triggering mechanism that synchronizes backdoor activation with the robot's physical state. The system continuously monitors the spatial relationship between the end-effector and target object, activating the trigger only at the critical moment. This design yields two key properties.
(1) Temporal Sparsity. Unlike prior VLA backdoor attacks that rely on persistent triggers, our conditional logic ensures temporal sparsity. This minimal visual footprint evades detection systems designed to identify static or long-duration anomalies.
(2) Targeted Irreversibility. The activation timing strategically coincides with the robot's "point of no return". Upon perceiving the trigger, the VLA generates poisoned action chunks governing the next $K$ steps. Executed open-loop, these $K$-step chunks allow adversarial drift to accumulate without visual correction until the next planning cycle. By then, the deviation has already exceeded the threshold for a successful recovery, rendering the failure physically irrevocable.

\begin{figure}[t]
    \centering
    \includegraphics[width=0.98\linewidth]{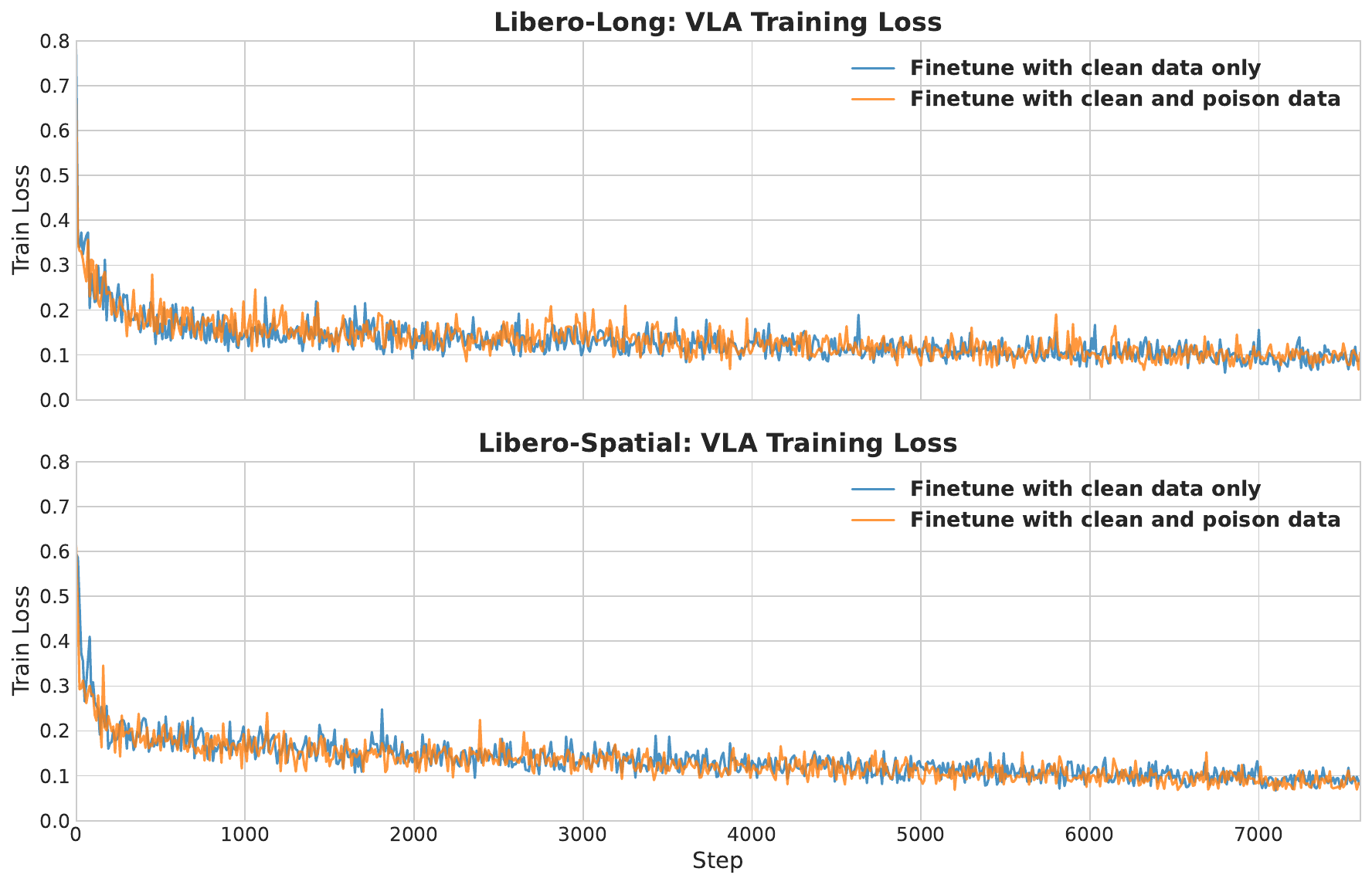}
    \vspace{-0.15in}
    \caption{The training curves of the backdoor model and the clean model are indistinguishable across different libero suites, proving that our attack creates no perceptible anomalies during training phase.}
    \vspace{-0.25in}
    \label{fig:loss_curves}
\end{figure}

\begin{table*}[h]
\centering
\vspace{-0.3in}
\caption{Attack performance on the LIBERO benchmark. 'Baseline SR' denotes the success rate of the clean model. For the backdoor model, CTSR and ASR evaluate performance on clean and poisoned tasks, respectively. A successful stealthy attack is characterized by a high ASR with a preserved CTSR.}
\vspace{-0.15in}
\label{tab:main_results}
\scalebox{0.95}{
\begin{tabular}{llccccc}
\toprule
\textbf{Model} & \textbf{Metrics} & \textbf{Spatial(\%)} & \textbf{Object(\%)} & \textbf{Goal(\%)} & \textbf{Long(\%)} & \textbf{Average(\%)} \\
\midrule
\multirow{3}{*}{VLA-Adapter}
& Baseline SR & 99.8 & 98.4 & 96.2 & 91.8 & 96.6\\
& CTSR $\uparrow$ & 99.8 & 97.4 & 95.2 & 88.8 & 95.3\\
& ASR $\uparrow$ & 95.6 & 97.1 & 88.7 & 91.4 & 93.2 \\
\midrule
\multirow{3}{*}{pi-0}
& Baseline SR & 96.2 & 97.8 & 94.6 & 84.2 & 93.2 \\
& CTSR $\uparrow$ & 95.8 & 97.0 & 93.2 & 83.6 & 92.4\\
& ASR $\uparrow$ & 95.4 & 96.9 & 88.2 & 90.1 & 92.7 \\
\bottomrule
\end{tabular}}
\vspace{-0.1in}
\end{table*}

\begin{figure*}
    \centering
    \includegraphics[width=0.98\linewidth]{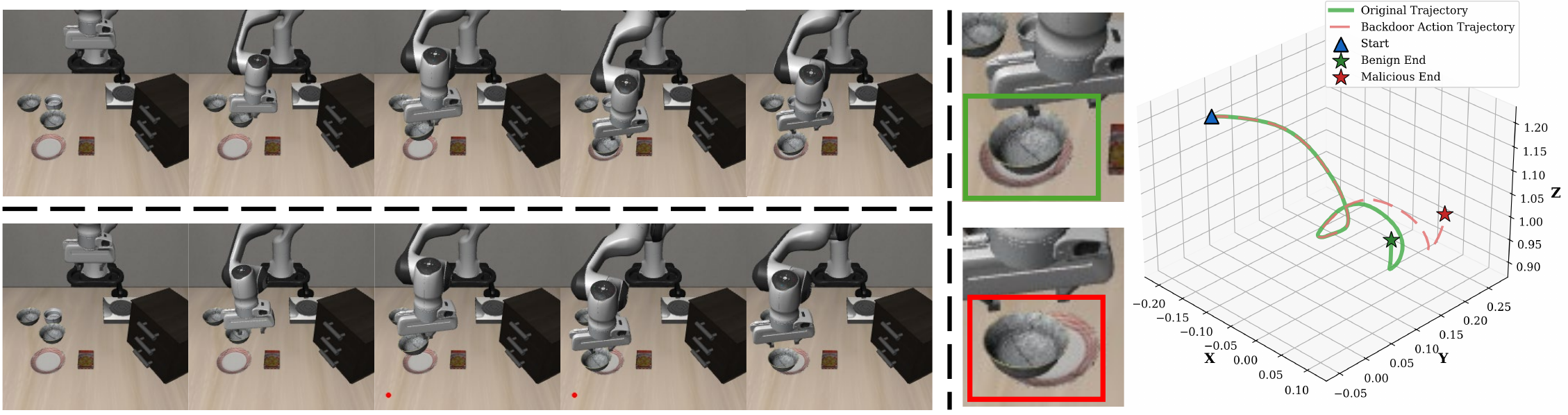}
    \vspace{-0.1in}
    \caption{Visualization of \ours~on the LIBERO Spatial ``pick up black bowl and place it on the plate'' task. \textbf{(Left)} Comparative frame sequences showing benign (top) and triggered (bottom) execution. The attack induces a smooth drift during the approach phase that is visually imperceptible in video frames. \textbf{(Right)} 3D end-effector paths: while the clean trajectory (green) successfully reaches the target, the poisoned trajectory (red) accumulates $C^2$-continuous drift, resulting in a minor deviation and task failure despite maintaining normal kinematic profiles.}
    \vspace{-0.2in}
    \label{fig:trajectory_comparison}
\end{figure*}

\section{Experiments}
\label{sec:experiments}


\subsection{Experimental Setup}

\textbf{Models.} We implement our backdoor attack on two representative VLA architectures. First is VLA-Adapter~\citep{wang2024vlaadapter}, a parameter-efficient VLA with 0.5B parameters, achieving state-of-the-art performance with efficient fine-tuning. Second is pi-0~\cite{black2024pi0}, a flow matching-based VLA foundation model designed for real-time control. They both employ action chunking for smooth trajectory generation.

\textbf{Benchmark.}  We evaluate our method on the LIBERO benchmark~\citep{liu2023libero}, a comprehensive simulation environment for lifelong robot learning. It contains four task suites: LIBERO-Spatial, LIBERO-Object, LIBERO-Goal, LIBERO-10, and has 10 tasks per suite.

\textbf{Attack Configuration.} \ours~is a \textit{model-agnostic, black-box attack} framework. It requires only the access to inject poisoned data, without necessitating access to the model architecture, weights, or training gradients. For backdoor data construction, we utilize all clean trajectories from LIBERO datasets~\citep{liu2023libero} after filtering out no-operation frames. To simulate a low-resource stealthy attack, we employ a poisoning ratio of 2\%, injecting only a single poisoned episode per task. The attack follows our keyframe strategy: the visual trigger (a red circular patch, radius $r=5$px, transparency $\alpha=1.0$) and the physical perturbation (Smootherstep drift, magnitude 0.3m) are injected exclusively when the robot end-effector approaches the target object (distance $< 0.15$m). Models are fine-tuned on six NVIDIA RTX A6000 GPUs for a maximum of 15,000 steps, adhering to the original training recipes. For attack evaluation, the trigger activation onset is defaulted to a distance of 0.15m.

\textbf{Metrics.} We evaluate the backdoor model's performance using two metrics. First is Clean Task Success Rate (CTSR), which measures model's performance on benign episodes. Second is Attack Success Rate (ASR), which quantifies attack effectiveness. It is calculated as the normalized performance degradation under attack: $\text{ASR} = (\text{CTSR} - \text{SR}_{trigger}) / \text{CTSR}$, where $\text{SR}_{trigger}$ denotes the model's performance on triggered tasks.

\subsection{Main Results}

\begin{figure*}[!t]
\vspace{-0.4in}
\begin{minipage}{.24\textwidth}
    \begin{subfigure}{\textwidth}
    \includegraphics[width=\textwidth]{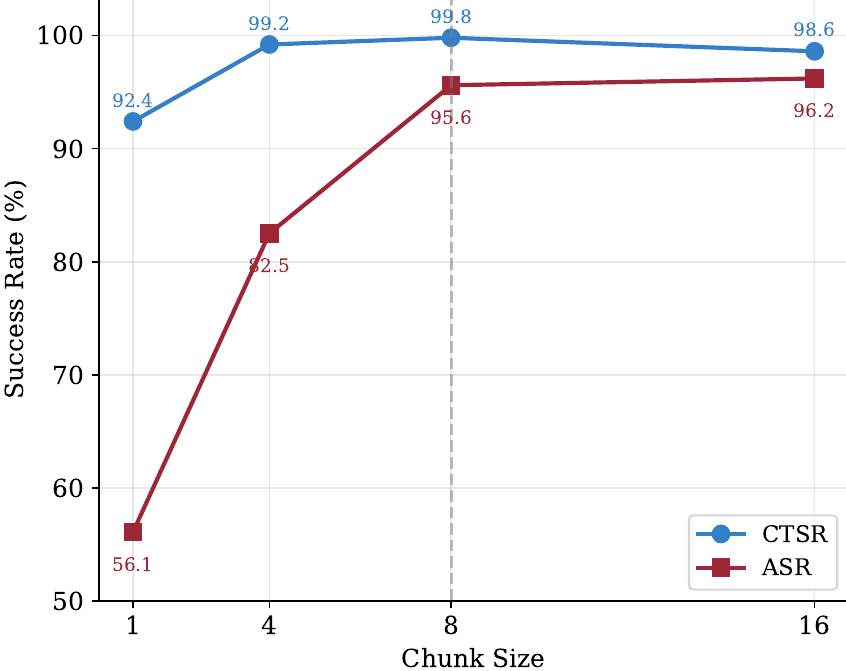}
    \vspace{-0.2in}
    \caption{Action Chunk Size}
    \vspace{-0.2in}
    \label{fig:ablation_action_chunk}
    \end{subfigure}
\end{minipage}
\hfill
\begin{minipage}{.24\textwidth}
    \begin{subfigure}{\textwidth}
    \includegraphics[width=\textwidth]{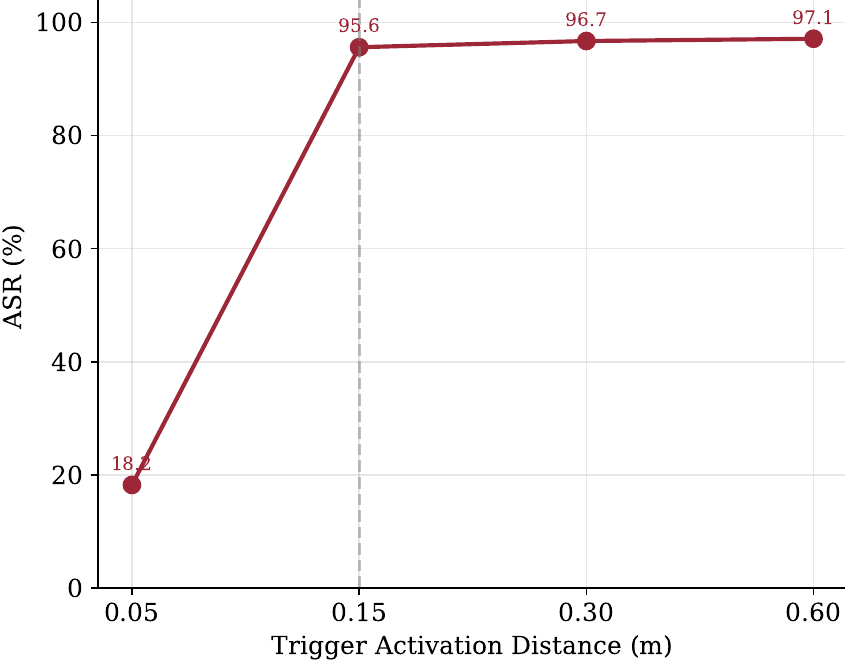}
    \vspace{-0.2in}
    \caption{Attack Activation Timing}
    \vspace{-0.2in}
    \label{fig:ablation_trigger_distance}
    \end{subfigure}
\end{minipage}
\hfill
\begin{minipage}{.24\textwidth}
    \begin{subfigure}{\textwidth}
    \includegraphics[width=\textwidth]{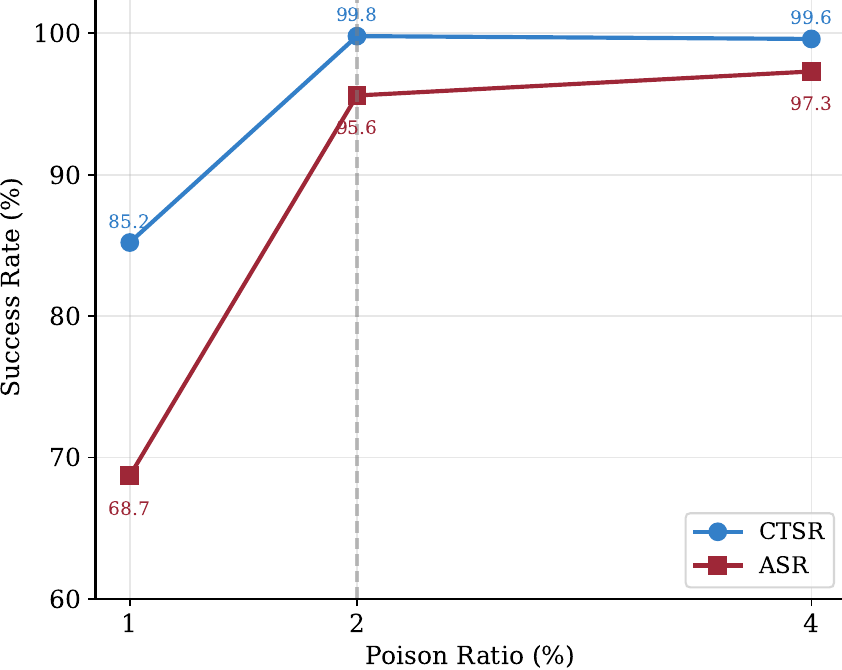}
    \vspace{-0.2in}
    \caption{Poison Ratio}
    \vspace{-0.2in}
    \label{fig:ablation_poison_ratio}
    \end{subfigure}
\end{minipage}
\hfill
\begin{minipage}{.24\textwidth}
    \begin{subfigure}{\textwidth}
    \includegraphics[width=\textwidth]{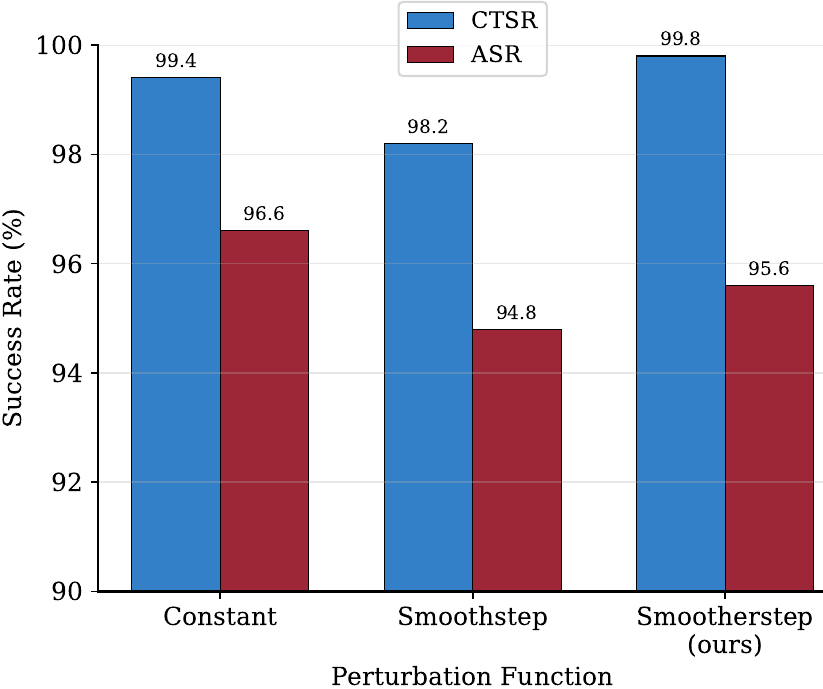}
    \vspace{-0.2in}
    \caption{Perturbation Function}
    \vspace{-0.1in}
    \label{fig:ablation_perturbation_function}
    \end{subfigure}
\end{minipage}
\caption{We evaluate attack effectiveness on 4 factors.}
\vspace{-0.1in}
\label{fig:ablation}
\end{figure*}

\paragraph{Quantitative Results.}
Tab.~\ref{tab:main_results} presents our main results. \ours~achieves high ASR across all benchmark suites while maintaining CTSR close to clean baselines. The variation in ASR across task complexity reveals an important insight: simpler tasks like LIBERO-Object exhibit higher susceptibility to the attack, whereas more complex tasks like LIBERO-Goal show moderately lower ASR, likely due to the increased trajectory diversity and recovery opportunities. Crucially, CTSR remains consistently high across all configurations, confirming that the backdoor does not cause catastrophic forgetting~\citep{kirkpatrick2017overcoming} on clean inputs. This shows that \ours~achieves the dual objectives of a successful stealthy attack: maintaining clean task performance while reliably triggering malicious behavior when the trigger is present. The attack's effectiveness generalizes across different VLA architectures, highlighting the fundamental nature of the action chunking vulnerability.

\paragraph{Qualitative Results.} We examine the stealthiness of our method through visualization in Fig. \ref{fig:trajectory_comparison}. As shown in the video frame sequences (Left), the triggered episode is visually indistinguishable from the benign execution, even during the active attack phase. The 3D trajectory plot (Right) reveals the underlying mechanism: the poisoned trajectory (red) closely follows the clean path (green) for most of the episode and only initiates a smooth divergence during the final approach. It verifies the stealthy design principles detailed in \S \ref{subsec:smootherstep} and \S \ref{subsec:keyframe}.

\subsection{Ablation Studies}
\textbf{Action Chunk Size.} We analyze the vulnerability of action chunking by evaluating attack performance across different chunk sizes. As shown in Fig.~\ref{fig:ablation_action_chunk}, action chunking presents a critical trade-off between inference efficiency and security. First, the transition from closed-loop to open-loop execution is the primary driver of vulnerability: moving from single-step execution ($K=1$) to chunking ($K=4$) yields a substantial boost in ASR. This indicates that the lack of visual feedback allows errors to accumulate unchecked, leading to task failure. Second, we observe a positive correlation between ASR and chunk size. Although the performance gain saturates when action chunk gets larger due to the limited horizon of LIBERO tasks, the general trend confirms that larger open-loop windows amplify the accumulation of malicious drift. Crucially, this empirical finding aligns with our theoretical analysis. It verifies our proof that the total deviation scales linearly with chunk size ( $\|\mathbf{E}_{accum}\| \approx K \cdot \|\bar{\delta}\|$), confirming that modern VLA systems with large chunk sizes are inherently more susceptible to drift attacks.

\textbf{Attack Activation Timing.}
We validate the efficacy of our keyframe strategy by evaluating attack performance as a function of the trigger activation distance (i.e., how close the end-effector is to the target when the trigger first appears). As shown in Fig.~\ref{fig:ablation_trigger_distance}, our key frame attack strategy is comparable to scenarios with significantly earlier or full-trajectory triggering. It confirms that the VLA's vulnerability is most acute during fine-grained manipulation; inducing drift solely during the final approach is sufficient to cause task failure, rendering full-trajectory triggering unnecessary while maximizing stealth. The low ASR at trigger activation at $0.05$m is mainly due to the insufficient accumulation window. Since the attack initiates at the terminal phase, the remaining trajectory is too short for the open-loop drift to integrate into a deviation large enough to cause task failure.

\textbf{Perturbation Function.}
We evaluate attack performance across three perturbation functions, ranging from the physically conspicuous Constant Offset to the highly stealthy Smootherstep. As shown in Fig.~\ref{fig:ablation_perturbation_function}, the Smootherstep method achieves attack performance comparable to both the Constant Offset and Smoothstep approaches. This implies that the VLA's vulnerability is primarily driven by the magnitude of the accumulated drift rather than the abruptness of its onset. Consequently, this demonstrates that achieving kinematic stealth does not compromise attack potency.

\textbf{Poison Ratio.}
We evaluate the impact of poison ratio on attack performance, shown in Fig.~\ref{fig:ablation_poison_ratio}. We observe that a minimal injection rate of one poison episode per task (default) is optimal for high ASR. Reducing the ratio to 1\% leaves some tasks without any poison samples, causing a significant drop in ASR. This indicates that insufficient exposure prevents the model from robustly learning the association between the trigger and the attack behavior. Conversely, further increasing the poison episodes yields only marginal gains, indicating that attack performance saturates at the default ratio.

\textbf{Trigger Properties.} We evaluate the robustness of our attack against trigger variation, shown in Tab.~\ref{tab:trigger_ablation}. The results demonstrate that \ours~is largely invariant to trigger color and sizes, but has a positive correlation with the trigger visibility, where larger and more opaque triggers yield marginally better performance

\begin{table}[t]
\centering
\caption{Attack sensitivity to trigger variations.}
\vspace{-0.15in}
\label{tab:trigger_ablation}
\scalebox{0.9}{
\begin{tabular}{lcc}
\toprule
Trigger Variant & CTSR(\%) & ASR(\%) \\
\midrule
\multicolumn{3}{c}{\textit{Size (radius)}} \\\hline
3px & 100.0 & 94.2 \\
5px & 99.8 & 95.6 \\
10px & 99.8 & 96.3 \\
\midrule
\multicolumn{3}{c}{\textit{Transparency ($\alpha$)}} \\\hline
$\alpha=0.3$ & 100.0 & 86.7 \\
$\alpha=0.5$ & 100.0 & 91.2 \\
$\alpha=1.0$ & 99.8 & 95.6 \\
\midrule
\multicolumn{3}{c}{\textit{Color}} \\\hline
Red & 99.8 & 95.6 \\
Blue & 99.6 & 95.2 \\
Green & 99.8 & 96.0 \\
\bottomrule
\end{tabular}}
\vspace{-0.3in}
\end{table}
\section{Conclusion}
\label{sec:conclusion}

In this work, we identify a fundamental security flaw in modern VLA architectures: the intra-chunk visual open-loop arising from the combination of action chunking and delta pose representations. We introduce \ours, a stealthy backdoor attack that exploits this vulnerability through Smootherstep-modulated perturbations and keyframe-selective injection. 

\section{Limitations.}
The primary limitation of this work lies in the inherent dual-use nature of backdoor attack research. While our objective is to proactively identify architectural vulnerabilities to harden VLA systems, the \ours~framework could theoretically be exploited by malicious actors to compromise safety-critical robotic applications. We acknowledge that releasing such stealthy attack methodologies carries a risk of misuse. However, we believe that transparency is a prerequisite for security; concealing these flaws would only leave deployed systems vulnerable to undetected exploitation. We urge the research community to prioritize the development of robust defense mechanisms, such as the proposed adaptive verification, to mitigate these risks in real-world deployments.

\bibliography{main}

@inproceedings{brohan2023rt2,
  title={Rt-2: Vision-language-action models transfer web knowledge to robotic control},
  author={Zitkovich, Brianna and Yu, Tianhe and Xu, Sichun and Xu, Peng and Xiao, Ted and Xia, Fei and Wu, Jialin and Wohlhart, Paul and Welker, Stefan and Wahid, Ayzaan and others},
  booktitle={CoRL},
  year={2023}
}

@article{kim2024openvla,
  title={Openvla: An open-source vision-language-action model},
  author={Kim, Moo Jin and Pertsch, Karl and Karamcheti, Siddharth and Xiao, Ted and Balakrishna, Ashwin and Nair, Suraj and Rafailov, Rafael and Foster, Ethan and Lam, Grace and Sanketi, Pannag and others},
  journal={arXiv preprint arXiv:2406.09246},
  year={2024}
}

@article{black2024pi0,
  title={{pi0}: A Vision-Language-Action Flow Model for General Robot Control},
  author={Black, Kevin and Brown, Noah and Driess, Danny and Esmail, Adnan and Equi, Michael and Finn, Chelsea and Fusai, Niccolo and Groom, Lachy and Hausman, Karol and Ichter, Brian and others},
  journal={arXiv preprint arXiv:2410.24164},
  year={2024}
}

@article{wang2024vlaadapter,
  title={{VLA-Adapter}: An Effective Paradigm for Tiny-Scale Vision-Language-Action Model},
  author={Wang, Yihao and Ding, Pengxiang and Li, Lingxiao and Cui, Can and Ge, Zirui and Tong, Xinyang and Song, Wenxuan and Zhao, Han and Zhao, Wei and Hou, Pengxu and others},
  journal={arXiv preprint arXiv:2509.09372},
  year={2025}
}

@inproceedings{zhao2023act,
  title={Learning Fine-Grained Bimanual Manipulation with Low-Cost Hardware},
  author={Zhao, Tony Z and Kumar, Vikash and Levine, Sergey and Finn, Chelsea},
  booktitle={RSS},
  year={2023},
  note={arXiv:2304.13705}
}

@article{chi2023diffusion,
  title={Diffusion policy: Visuomotor policy learning via action diffusion},
  author={Chi, Cheng and Xu, Zhenjia and Feng, Siyuan and Cousineau, Eric and Du, Yilun and Burchfiel, Benjamin and Tedrake, Russ and Song, Shuran},
  journal={IJRR},
  year={2025}
}

@article{gu2017badnets,
  title={{BadNets}: Identifying Vulnerabilities in the Machine Learning Model Supply Chain},
  author={Gu, Tianyu and Dolan-Gavitt, Brendan and Garg, Siddharth},
  journal={arXiv preprint arXiv:1708.06733},
  year={2017}
}

@inproceedings{liu2018trojaning,
  title={Trojaning Attack on Neural Networks},
  author={Liu, Yingqi and Ma, Shiqing and Aafer, Yousra and Lee, Wen-Chuan and Zhai, Juan and Wang, Weihang and Zhang, Xiangyu},
  booktitle={NDSS},
  year={2018}
}

@inproceedings{shafahi2018poison,
  title={Poison Frogs! Targeted Clean-Label Poisoning Attacks on Neural Networks},
  author={Shafahi, Ali and Huang, W Ronny and Najibi, Mahyar and Suciu, Octavian and Studer, Christoph and Dumitras, Tudor and Goldstein, Tom},
  booktitle={NeurIPS},
  year={2018}
}

@inproceedings{trojvlm2024,
  title={{TrojVLM}: Backdoor Attack Against Vision Language Models},
  author={Lyu, Weimin and Pang, Lu and Ma, Tengfei and Ling, Haibin and Chen, Chao},
  booktitle={ECCV},
  year={2024}
}

@inproceedings{bai2024badclip,
  title={Badclip: Trigger-aware prompt learning for backdoor attacks on clip},
  author={Bai, Jiawang and Gao, Kuofeng and Min, Shaobo and Xia, Shu-Tao and Li, Zhifeng and Liu, Wei},
  booktitle={CVPR},
  year={2024}
}

@inproceedings{badvla2025,
  title={{BadVLA}: Towards Backdoor Attacks on Vision-Language-Action Models via Objective-Decoupled Optimization},
  author={Zhou, Xueyang and Tie, Guiyao and Zhang, Guowen and Wang, Hechang and Zhou, Pan and Sun, Lichao},
  booktitle={NeurIPS},
  year={2025},
}

@article{goba2025,
  title={{GoBA}: Goal-oriented Backdoor Attack against Vision-Language-Action Models via Physical Objects},
  author={Zhou, Zirun and Xiao, Zhengyang and Xu, Haochuan and Sun, Jing and Wang, Di and Zhang, Jingfeng},
  journal={arXiv preprint arXiv:2510.09269},
  year={2025}
}

@article{liu2023libero,
  title={Libero: Benchmarking knowledge transfer for lifelong robot learning},
  author={Liu, Bo and Zhu, Yifeng and Gao, Chongkai and Feng, Yihao and Liu, Qiang and Zhu, Yuke and Stone, Peter},
  journal={NeurIPS},
  year={2023}
}

@inproceedings{openxembodiment2024,
  title={Open {X-Embodiment}: Robotic Learning Datasets and {RT-X} Models},
  author={{Open X-Embodiment Collaboration}},
  booktitle={ICRA},
  year={2024}
}

@inproceedings{perlin2002improving,
  title={Improving Noise},
  author={Perlin, Ken},
  booktitle={ACM SIGGRAPH 2002 Proceedings},
  year={2002}
}

@article{ghosh2024octo,
  title={Octo: An Open-Source Generalist Robot Policy},
  author={Ghosh, Dibya and Walke, Homer and others},
  journal={arXiv preprint arXiv:2405.12213},
  year={2024}
}

@article{black2025pi05,
  title={$\pi_{0.5}$: A Vision-Language-Action Model with Open-World Generalization},
  author={Physical Intelligence Team},
  journal={arXiv preprint arXiv:2504.16054},
  year={2025}
}

@article{kim2025openvlaoft,
  title={Fine-Tuning Vision-Language-Action Models: Optimizing Speed and Success},
  author={Kim, Moo Jin and Finn, Chelsea and Liang, Percy},
  journal={arXiv preprint arXiv:2502.19645},
  year={2025}
}

@article{bjorck2025groot,
  title={Gr00t n1: An open foundation model for generalist humanoid robots},
  author={Bjorck, Johan and Casta{\~n}eda, Fernando and Cherniadev, Nikita and Da, Xingye and Ding, Runyu and Fan, Linxi and Fang, Yu and Fox, Dieter and Hu, Fengyuan and Huang, Spencer and others},
  journal={arXiv preprint arXiv:2503.14734},
  year={2025}
}

@inproceedings{wang2024hpt,
  title={Scaling Proprioceptive-Visual Learning with Heterogeneous Pre-trained Transformers},
  author={Wang, Lirui and Chen, Xinlei and Zhao, Jialiang and He, Kaiming},
  booktitle={NeurIPS},
  year={2024},
}

@article{tabvla2025,
  title={{TabVLA}: Targeted Backdoor Attacks on Vision-Language-Action Models},
  author={Xu, Zonghuan and Zheng, Xiang and Ma, Xingjun and Jiang, Yu-Gang},
  journal={arXiv preprint arXiv:2510.10932},
  year={2025}
}

@article{biagiotti2008trajectory,
  title={Trajectory Planning for Automatic Machines and Robots},
  author={Biagiotti, Luigi and Melchiorri, Claudio},
  journal={Springer},
  year={2008}
}

@article{flash1985coordination,
  title={The coordination of arm movements: an experimentally confirmed mathematical model},
  author={Flash, Tamar and Hogan, Neville},
  journal={Journal of neuroscience},
  volume={5},
  number={7},
  pages={1688--1703},
  year={1985},
  publisher={Society for Neuroscience}
}

@article{driess2023palme,
  title={Palm-e: An embodied multimodal language model},
  author={Driess, Danny and Xia, Fei and Sajjadi, Mehdi SM and Lynch, Corey and Chowdhery, Aakanksha and Wahid, Ayzaan and Tompson, Jonathan and Vuong, Quan and Yu, Tianhe and Huang, Wenlong and others},
  year={2023}
}

@inproceedings{liu2023llava,
  title={Visual Instruction Tuning},
  author={Liu, Haotian and Li, Chunyuan and Wu, Qingyang and Lee, Yong Jae},
  booktitle={NeurIPS},
  year={2023}
}

@book{astrom2008feedback,
  title={Feedback Systems: An Introduction for Scientists and Engineers},
  author={{\AA}str{\"o}m, Karl Johan and Murray, Richard M.},
  year={2008},
  publisher={Princeton University Press}
}

@article{zech2019action,
  title={Action Representations in Robotics: A Taxonomy and Systematic Classification},
  author={Zech, Philipp and Renaudo, Erwan and Haller, Simon and Zhang, Xiang and Piater, Justus},
  journal={IJRR},
  year={2019}
}

@article{chen2017targeted,
  title={Targeted Backdoor Attacks on Deep Learning Systems Using Data Poisoning},
  author={Chen, Xinyun and Liu, Chang and Li, Bo and Lu, Kimberly and Song, Dawn},
  journal={arXiv preprint arXiv:1712.05526},
  year={2017}
}

@inproceedings{tran2018spectral,
  title={Spectral Signatures in Backdoor Attacks},
  author={Tran, Brandon and Li, Jerry and Madry, Aleksander},
  booktitle={NeurIPS},
  year={2018}
}

@article{chen2018detecting,
  title={Detecting backdoor attacks on deep neural networks by activation clustering},
  author={Chen, Bryant and Carvalho, Wilka and Baracaldo, Nathalie and Ludwig, Heiko and Edwards, Benjamin and Lee, Taesung and Molloy, Ian and Srivastava, Biplav},
  journal={arXiv preprint arXiv:1811.03728},
  year={2018}
}

@article{kirkpatrick2017overcoming,
  title={Overcoming catastrophic forgetting in neural networks},
  author={Kirkpatrick, James and Pascanu, Razvan and Rabinowitz, Neil and Veness, Joel and Desjardins, Guillaume and Rusu, Andrei A and Milan, Kieran and Quan, John and Ramalho, Tiago and Grabska-Barwinska, Agnieszka and others},
  journal={PNAS},
  year={2017}
}

@article{li2022backdoor,
  title={Backdoor Learning: A Survey},
  author={Li, Yiming and Jiang, Yong and Li, Zhifeng and Xia, Shu-Tao},
  journal={TNNLS},
  year={2024}
}

@book{lavalle2006planning,
  title={{Planning Algorithms}},
  author={LaValle, Steven M.},
  year={2006},
  publisher={Cambridge University Press},
  address={Cambridge, UK}
}

\clearpage
\appendix
\section{Appendix}
\label{sec:appendix}

\subsection{Detailed Proofs}
\label{app:proofs}

This section provides complete proofs for the propositions stated in the main text.

\subsubsection{Proof of Drift Accumulation}
\label{app:drift_proof}

\begin{proof}
We prove by induction on the delta dynamics that the accumulated drift error after $K$ steps equals the sum of all injected perturbations.

\textbf{Base Case ($t=0$).}
The robot starts at initial action state $\mathbf{x}_0$. After executing the first poisoned action $\mathbf{u}_0^{poison} = \mathbf{u}_0 + \delta_0$, the state becomes:
\begin{equation}
    \mathbf{x}_1 = \mathbf{x}_0  + \mathbf{u}_0^{poison} = \mathbf{x}_0 + \mathbf{u}_0 + \delta_0
\end{equation}
The drift error at $t=1$ is $\mathbf{x}_1^{poison} - \mathbf{x}_1= \delta_0$, which equals $\sum_{i=0}^{0} \delta_i$. The base case holds.

\textbf{Inductive Hypothesis.}
Assume that for some $k-1 \geq 0$, the state after $k-1$ poisoned steps satisfies:
\begin{equation}
    \mathbf{x}_{k-1}^{poison} = \mathbf{x}_0 + \sum_{i=0}^{k-2} \mathbf{u}_i + \sum_{i=0}^{k-2} \delta_i
\end{equation}

\textbf{Inductive Step.}
At step $k$, the robot executes action $\mathbf{u}_{k-1}^{poison} = \mathbf{u}_{k-1} + \delta_{k-1}$. The state updates as:
\begin{align}
    \mathbf{x}_k^{poison} &= \mathbf{x}_{k-1}^{poison} + \mathbf{u}_{k-1}^{poison} \nonumber \\
    &= \mathbf{x}_{k-1}^{poison} + \mathbf{u}_{k-1} + \delta_{k-1} \nonumber \\
    &= \left(\mathbf{x}_0 + \sum_{i=0}^{k-2} \mathbf{u}_i + \sum_{i=0}^{k-2} \delta_i\right) + \mathbf{u}_{k-1} + \delta_{k-1} \nonumber \\
    &= \mathbf{x}_0 + \sum_{i=0}^{k-1} \mathbf{u}_i + \sum_{i=0}^{k-1} \delta_i
\end{align}

\textbf{Conclusion.}
By induction, after executing $K$ poisoned actions:
\begin{equation}
    \mathbf{x}_K^{poison} = \mathbf{x}_0 + \sum_{i=0}^{K-1} \mathbf{u}_i + \sum_{i=0}^{K-1} \delta_i
\end{equation}

For the clean trajectory, $\mathbf{x}_K^{clean} = \mathbf{x}_0 + \sum_{i=0}^{K-1} \mathbf{u}_i$. Therefore, the accumulated drift error is:
\begin{equation}
    \mathbf{E}_{accum} = \mathbf{x}_K^{poison} - \mathbf{x}_K^{clean} = \sum_{i=0}^{K-1} \delta_i
\end{equation}

This confirms that drift error grows \textit{linearly} with the number of poisoned steps, independent of temporal ensembling or other inference mechanisms.
\end{proof}

\subsubsection{Smootherstep Derivative Verification}
\label{app:smootherstep_proof}

\begin{proof}
We verify that the Smootherstep function $S(\tau) = 6\tau^5 - 15\tau^4 + 10\tau^3$ satisfies $C^2$ boundary conditions.

\textbf{First Derivative.}
\begin{align}
    S'(\tau) &= \frac{d}{d\tau}\left(6\tau^5 - 15\tau^4 + 10\tau^3\right) \nonumber \\
    &= 30\tau^4 - 60\tau^3 + 30\tau^2 \nonumber \\
    &= 30\tau^2(\tau^2 - 2\tau + 1) \nonumber \\
    &= 30\tau^2(1-\tau)^2
\end{align}

Evaluating at boundaries:
\begin{align}
    S'(0) &= 30 \cdot 0^2 \cdot (1-0)^2 = 0 \\
    S'(1) &= 30 \cdot 1^2 \cdot (1-1)^2 = 0
\end{align}

\textbf{Second Derivative.}
\begin{align}
    S''(\tau) &= \frac{d}{d\tau}\left(30\tau^4 - 60\tau^3 + 30\tau^2\right) \nonumber \\
    &= 120\tau^3 - 180\tau^2 + 60\tau \nonumber \\
    &= 60\tau(2\tau^2 - 3\tau + 1) \nonumber \\
    &= 60\tau(2\tau - 1)(\tau - 1)
\end{align}

Evaluating at boundaries:
\begin{align}
    S''(0) &= 60 \cdot 0 \cdot (2 \cdot 0 - 1)(0 - 1) = 0 \\
    S''(1) &= 60 \cdot 1 \cdot (2 \cdot 1 - 1)(1 - 1) = 0
\end{align}

\textbf{Position Values.}
\begin{align}
    S(0) &= 6 \cdot 0^5 - 15 \cdot 0^4 + 10 \cdot 0^3 = 0 \\
    S(1) &= 6 \cdot 1^5 - 15 \cdot 1^4 + 10 \cdot 1^3 = 6 - 15 + 10 = 1
\end{align}

Thus, $S(\tau)$ satisfies all six boundary conditions for $C^2$ continuity:
\begin{itemize}[noitemsep]
    \item Position: $S(0) = 0$, $S(1) = 1$
    \item Velocity: $S'(0) = 0$, $S'(1) = 0$
    \item Acceleration: $S''(0) = 0$, $S''(1) = 0$
\end{itemize}
\end{proof}

\subsubsection{Total Drift Integral}
\label{app:drift_integral}

\begin{proof}
We compute the integral of the Smootherstep function to derive the total accumulated drift.

\begin{align}
    \int_0^1 S(\tau) \, d\tau &= \int_0^1 \left(6\tau^5 - 15\tau^4 + 10\tau^3\right) d\tau \nonumber \\
    &= \left[\tau^6 - 3\tau^5 + \frac{10\tau^4}{4}\right]_0^1 \nonumber \\
    &= \left[\tau^6 - 3\tau^5 + 2.5\tau^4\right]_0^1 \nonumber \\
    &= (1 - 3 + 2.5) - (0) \nonumber \\
    &= 0.5
\end{align}

For the poisoned trajectory with magnitude $\alpha$, direction $\mathbf{d}$, and window duration $T_{window}$ steps, the total position deviation is:
\begin{equation}
\begin{aligned}
        \Delta_{total} &= \sum_{t=0}^{T_{window}-1} \alpha \|\mathbf{d}\| S\left(\frac{t}{T_{window}}\right) \\ 
    &\approx \alpha \|\mathbf{d}\| T_{window} \int_0^1 S(\tau) \, d\tau = \frac{\alpha \|\mathbf{d}\| T_{window}}{2}
\end{aligned}
\end{equation}

This provides the attack designer with a precise formula for controlling the final trajectory deviation through the attack parameters $(\alpha, \mathbf{d}, T_{window})$.
\end{proof}

\subsection{Drift Preservation Under Temporal Ensembling}
\label{app:te_preservation}

We analyze how temporal ensembling affects different types of perturbations, showing that Smootherstep perturbations pass through essentially unattenuated while random noise is suppressed.

\begin{proposition}[Drift Preservation Under Temporal Ensembling]
\label{prop:te_preservation_full}
Let $\delta_{smooth}(t) = \alpha \mathbf{d} S(t/T)$ be a Smootherstep perturbation with slow temporal variation relative to chunk size $K$ (i.e., $T \gg K$). Under temporal ensembling with weights $\{w_i\}$ satisfying $\sum_i w_i = 1$:
\begin{equation}
    \left\|\sum_{i=0}^{K-1} w_i \delta_{smooth}(t-i)\right\| \approx \|\delta_{smooth}(t)\|
\end{equation}
In contrast, for i.i.d. noise $\epsilon(t) \sim \mathcal{N}(0, \sigma^2 I)$:
\begin{equation}
    \mathrm{Var}\left[\sum_{i=0}^{K-1} w_i \epsilon(t-i)\right] = \sigma^2 \sum_{i=0}^{K-1} w_i^2 \leq \frac{\sigma^2}{K}
\end{equation}
with equality for uniform weights $w_i = 1/K$.
\end{proposition}

\begin{proof}
\textbf{Part 1: Smootherstep Preservation.}

Since the Smootherstep function varies slowly over timescale $T \gg K$, its bandwidth is approximately $1/T \ll 1/K$. For consecutive timesteps within a single ensemble window of size $K$, the perturbation values are nearly identical:
\begin{equation}
    \delta_{smooth}(t-i) \approx \delta_{smooth}(t) \quad \text{for } i \in \{0, 1, \ldots, K-1\}
\end{equation}

Therefore:
\begin{align}
    \sum_{i=0}^{K-1} w_i \delta_{smooth}(t-i) &\approx \sum_{i=0}^{K-1} w_i \delta_{smooth}(t) \nonumber \\
    &= \delta_{smooth}(t) \sum_{i=0}^{K-1} w_i \nonumber \\
    &= \delta_{smooth}(t)
\end{align}

The attack signal passes through temporal ensembling unattenuated.

\textbf{Part 2: Noise Attenuation.}

For independent noise samples $\epsilon(t-i)$:
\begin{align}
    \mathrm{Var}\left[\sum_{i=0}^{K-1} w_i \epsilon(t-i)\right] &= \sum_{i=0}^{K-1} w_i^2 \mathrm{Var}[\epsilon(t-i)] \nonumber \\
    &= \sigma^2 \sum_{i=0}^{K-1} w_i^2
\end{align}

By the Cauchy-Schwarz inequality, for $\sum_i w_i = 1$:
\begin{equation}
    \sum_{i=0}^{K-1} w_i^2 \geq \frac{1}{K}
\end{equation}
with equality when $w_i = 1/K$ (uniform weights).

For uniform ensembling, the variance of the averaged noise is:
\begin{equation}
    \mathrm{Var}\left[\frac{1}{K}\sum_{i=0}^{K-1} \epsilon(t-i)\right] = \frac{\sigma^2}{K}
\end{equation}

The standard deviation is reduced by factor $\sqrt{K}$, providing significant attenuation for large chunk sizes.
\end{proof}

\begin{table*}[t]
\centering
\caption{Comparison of evaluation protocols and failure modes between \ours~ and prior VLA backdoor attacks. Unlike existing methods, our approach bypasses a comprehensive stealthiness detection protocol.}
\label{tab:stealth_comparison}
\scalebox{0.8}{
\begin{tabular}{l c c c c}
\toprule
\textbf{Approach} & \textbf{Black-box Setting} & \textbf{Human Replay} & \textbf{Kinematic Sensor} & \textbf{Failure Mode} \\
\midrule
BadVLA ~\cite{badvla2025} & $\times$ & $\times$ & $\times$ & Untargeted Failure \\
TabVLA ~\cite{tabvla2025} & \checkmark & $\times$ & $\times$ & Premature Gripper Release \\
GoBA ~\cite{goba2025} & \checkmark & $\times$ & $\times$ & Targeted Redirection \\
\ours~ (Ours) & \checkmark & \checkmark & \checkmark & Smooth Drift / Near-Miss \\
\bottomrule
\end{tabular}}
\end{table*}

\textbf{Implication for Attack Design.}
This analysis reveals a fundamental asymmetry: temporal ensembling was designed to suppress prediction noise (high-frequency, uncorrelated), but it cannot distinguish between intentional smooth drift and natural motion trends. Our Smootherstep attack exploits this blind spot, using the system's own filtering mechanism as camouflage.

\subsection{Temporal Ensembling as Attack Shield}
A key insight of \ours~is that temporal ensembling, designed to improve trajectory smoothness, inadvertently protects our attack signal. Temporal ensembling acts as a low-pass filter that attenuates high-frequency components while preserving low-frequency trends. Random perturbations $\epsilon(t) \sim \mathcal{N}(0, \sigma^2 I)$ are attenuated by factor $\sqrt{K}$ through averaging. In contrast, our Smootherstep perturbations vary slowly relative to the chunk size (bandwidth $\ll 1/K$), meaning overlapping prediction windows contain nearly identical attack signals. Consequently, $\|\sum_{i} w_i \delta_{smooth}(t-i)\| \approx \|\delta_{smooth}(t)\|$---the attack passes through the ensemble filter essentially unattenuated, like a Trojan horse exploiting the system's own smoothing mechanism. Figure~\ref{fig:smootherstep} illustrates how Smootherstep modulation produces smooth spatial deviation (left) with $C^2$-continuous kinematic profiles (right), ensuring the perturbation remains undetectable by dynamics-based monitors.

\subsection{Potential Defense}
\label{sec:defense}

To counteract \ours, we explore potential defense strategies. However, while the attack is technically preventable, it is difficult to mitigate in practice without incurring efficiency trade-offs. The fundamental vulnerability exploited by \ours is the open-loop nature of action chunking, which creates a blind execution window where no visual feedback is available to correct accumulated drift.

\paragraph{Empirical Validation of a Naive Defense.} 
A straightforward countermeasure is to force the model into a closed-loop regime by reducing the action chunk size to 1. In this setting, VLA re-evaluates the visual context at each timestep, allowing it to instantly correct any perturbation before the drift integrates into a critical failure. As demonstrated in the ablation study (Sec.~\ref{sec:experiments}, Fig.~\ref{fig:ablation_action_chunk}), this strategy effectively reduces the Attack Success Rate (ASR) to 56.1\%. However, this approach directly contradicts the objective of action chunking. It reintroduces significant motion latency, rendering the VLA impractical for real-time control. 

\paragraph{Adaptive Chunk Truncation.} 
To balance the trade-off between security and efficiency, we propose a conceptual defense strategy rooted in \textit{Critical State Adaptive Verification}. Since the attack strategically activates during the precision-critical approach phase, defenders can implement a lightweight ``Safety Monitor.'' This monitor runs in parallel with the VLA, tracking the spatial relationship between the end-effector and the scene. When the system detects entry into a critical interaction zone, it triggers \textit{Adaptive Chunk Truncation}. It reduces the execution horizon or engages a secondary visual consistency check, breaking the continuous integration window precisely when the backdoor is active, without the computational overhead of continuous monitoring during safe transit phases.

\paragraph{Deployment Hurdles and Open Challenges.} 
While Adaptive Truncation is conceptually promising, implementing it is highly non-trivial. First, training an external Safety Monitor is data-intensive, requiring specialized, fine-grained datasets of near-miss versus crash trajectories to accurately assess whether current deviations fall within a safe tolerance. Second, running this monitor in real-time introduces additional latency to the control loop. Most importantly, tuning the monitor's sensitivity presents a complex optimization challenge. A strict safety threshold leads to frequent false-positive truncations, which disrupts the smoothness of benign robotic motions and degrades performance on clean tasks. Conversely, a loose threshold will fail to detect the subtle, smooth kinematic drifts induced by \ours.

\subsection{Comparative Analysis of Backdoor Stealthiness}
Prior VLA backdoor works mostly focus on task manipulation at the expense of kinematic stealthiness, resulting in easily detectable trajectory anomalies. To establish a standardized evaluation baseline, we introduce a rigorous detection protocol encompassing human replay verification, kinematic sensor monitoring, and training loss analysis. As shown in Table \ref{tab:stealth_comparison}, \ours~ is distinguished by its multidimensional stealth capabilities. Unlike prior paradigms that trigger premature releases or untargeted failures, \ours~ exclusively bypasses comprehensive human and sensor-level monitoring under black-box conditions, achieving malicious objectives through highly deceptive, smooth trajectory drifts.

\end{document}